\documentclass[10pt,aps,twocolumn,longbibliography,footnoteemail]{revtex4-1}

\usepackage{amsmath}
\usepackage{amsthm}
\usepackage{amsfonts}
\usepackage{amssymb}
\usepackage{mathrsfs}

\usepackage{graphicx}
\usepackage{xcolor}

\usepackage{hyperref} 
\hypersetup{pdfborder={0 0 0},colorlinks=true}

\theoremstyle{definition}

\theoremstyle{theorem}
\newtheorem{lemma}{Lemma}
\newtheorem*{lemma*}{Lemma}
\newtheorem{proposition}{Proposition}
\newtheorem*{proposition*}{Proposition}

\newtheorem*{conjecture*}{Conjecture}

\newtheorem*{corollary*}{Corollary}
\theoremstyle{remark}

\newcommand{\RR}{\mathbb{R}}

\newcommand{\U}{\mathscr{U}}

\newcommand{\B}{\mathscr{B}}
\newcommand{\C}{\mathscr{C}}

\renewcommand{\H}{\mathscr{H}}
\renewcommand{\P}{\mathscr{P}}
\renewcommand{\S}{\mathscr{S}}

\newcommand{\M}{\mathscr{M}}
\newcommand{\II}{\mathbb{I}}

\newcommand{\ket}[1]{\left| #1 \right\rangle}
\newcommand{\bra}[1]{\left\langle #1 \right|}

\newcommand{\ketbra}[2]{\ket{#1}\bra{#2}}

\renewcommand{\vec}[1]{\pmb{#1}}
\newcommand{\n}{\pmb{n}}

\newcommand{\abs}[1]{\left| #1\right|}

\renewcommand{\d}{\operatorname{d}}
\newcommand{\Tr}{\operatorname{Tr}}

\newcommand{\cell}{\operatorname{box}}
\newcommand{\cone}{\operatorname{cone}}

\begin{document}
\title{Non-separability and steerability of two-qubit states from the geometry of steering~outcomes}
\author{H. Chau Nguyen}
\email{chau@pks.mpg.de}
\affiliation{Max-Planck-Institut f\"ur Physik komplexer Systeme, \\ N\"othnitzer Stra{\ss}e 38, D-01187 Dresden, Germany}
\author{Thanh Vu}
\email{tvu@unl.edu}
\affiliation{Department of Mathematics, University of Nebraska-Lincoln, \\Lincoln, NE 68588, USA}
\begin{abstract}
When two qubits A and B are in an appropriate state, Alice can remotely steer Bob's system B into different ensembles by making different measurements on A. This famous phenomenon is known as quantum steering, or Einstein-Podolsky-Rosen steering. Importantly, quantum steering establishes the correspondence not only between a measurement on A (made by Alice) and an ensemble of B (owned by Bob) but also between each of Alice's measurement outcomes and an unnormalized conditional state of Bob's system. The unnormalized conditional states of B corresponding to all possible measurement outcomes of Alice are called Alice's steering outcomes. We show that, surprisingly, the $4$-dimensional geometry of Alice's steering outcomes completely determines both the non-separability of the two-qubit state and its steerability from her side. Consequently, the problem of classifying two-qubit states into non-separable and steerable classes is equivalent to geometrically classifying certain $4$-dimensional skewed double-cones. 
\end{abstract}

\maketitle

\section*{Introduction}
Quantum steering, or Einstein-Podolsky-Rosen (EPR) steering, arose from the first discussion on the non-local nature of quantum mechanics~\cite{Einstein1935a,Schrodinger1935a}. Subsequent attempts to clarify the notion of non-locality has led to the discovery of different classes of quantum non-locality.  Although Bell non-locality~\cite{Bell1964a} and non-separability~\cite{Werner1989a} were discussed rather early as two distinct classes of quantum non-locality, only recently was steerability defined~\cite{Wiseman2007a,Jones2007b}. 

One of the key concepts to define quantum steerability is the \emph{assemblage}. An assemblage is a set of ensembles that give rise to the same quantum state. Now consider the case where Alice and Bob share a bipartite system AB. Although the no-signalling theorem prevents Alice from affecting the reduced state of Bob's system B from distance~\cite{Popescu1994a,Pawlowski2009a}, she can remotely steer it into different ensembles by performing different measurements on her own system A~\cite{Einstein1935a,Schrodinger1935a}. These different ensembles of B form a certain assemblage--the steering assemblage. However, as Wiseman, Jones, and Doherty~\cite{Wiseman2007a} pointed out, if this steering assemblage is too restrictive in some sense, Alice may never convince Bob that she is actually steering his system remotely, in which case the state is unsteerable from her side. 

Shortly after steerability was defined, sufficient conditions for a state to be steerable were developed in terms of steering inequalities~\cite{Cavalcanti2009a,Saunders2011a,Zukowski2015a,Marciniak2015a,Kogia2015a,Zhu2015a}. Characterizing steerability has quickly gone beyond inequalities and multiple relationships to other concepts of quantum physics have been discovered. Steerability was shown to be equivalent to joint-measurability in~\cite{Uola2014a,Quintino2014a}.
Quantum steering in time was discussed by Chen \emph{et al.}~\cite{Chen2014a} and a close relationship between steerability and quantum-subchannel discrimination has been established~\cite{Piani2015a}. When steerability was realized as a resource for quantum information tasks, quantifying steerability naturally appeared as an important problem~\cite{Skrzypczyk2014a}. 

The concept of steerability is based on the correspondence between a measurement of Alice and an ensemble of Bob's system. However, quantum steering establishes a more elementary and much simpler correspondence: each of Alice's  measurement outcomes corresponds to a conditional state of Bob's system. When Alice gets a particular measurement outcome, the unnormalized conditional state of B is determined regardless of which measurement that particular outcome belongs to. The unnormalized conditional states of Bob's system corresponding to all possible measurement outcomes of Alice are referred to as Alice's steering outcomes. 
We show that, surprisingly, the $4$-dimensional (4D) geometry of Alice's steering outcomes completely determines both the non-separability of the two-qubit state~\cite{Werner1989a} and its steerability from her side~\cite{Wiseman2007a}; see Propositions~\ref{pros:separability} and~\ref{pros:steerability}. Thus, the problem of classifying two-qubit states into non-separable and steerable classes is conceptually simplified to classifying certain 4D skewed double-cones.

Although most of the definitions and many statements in this paper can be naturally generalized to higher dimensional systems, there are certain aspects of two-qubit systems that make the statements particularly simple and transparent. We thus restrict our analysis here to two-qubit systems and wish to discuss higher-dimensional ones elsewhere. 

\section*{EPR maps and steering outcomes}
Let us consider a qubit described by a $2$-dimensional (2D) Hilbert space. The hermitian operators acting on $\H$ with the Hilbert-Schmidt inner product $(A,B) \to \Tr (A^{\dagger}B)$ form a Euclidean space, denoted by $B^{H} (\H)$~\cite{Kadison1990a,Nielsen2010a}. Fix an orthogonal basis of $\H$. Letting $\sigma_0= \II$ be the identity matrix and $\{\sigma_i\}_{i=1}^3$ be the three Pauli matrices, then every hermitian operator $A$ acting on $\H$ can be written as
$A = \frac{1}{2} \sum_{i=0}^{3} X_i (A)  \sigma_i$, 
where $X_i (A)= \Tr (A\sigma_i)$. This coordinate system $\{X_i\}_{i=0}^3$ allows one to identify $B^H(\H)$ with the Euclidean space $\RR^4$. In particular, the zero operator $O$ is identified with $(0,0,0,0)^T$, and the identity operator $\II$ is identified with $(2,0,0,0)^T$.

Of particular interest to us is the cone of positive (semi-definite) operators $\M^{+}= \{M \vert 0 \le M\}$. In terms of the Euclidean coordinates, $\M^{+} = \{M \vert X_0(M) \ge 0, X_0(M)^2 \ge X_1 (M)^2 + X_2 (M)^2 + X_3 (M)^2 \}$, and thus is also called the forward light-cone at the origin $O$--a borrowed terminology from special relativity~\cite{Jevtic2014a}. Another object of interest to us is the set of \emph{measurement outcomes} $\M= \M^{+} \cap \M^{-}$, where $\M^{-}= \{M \vert \II \ge M\}$. It is easy to see that $\M^{-}$ is the backward light-cone at $\II$. Thus $\M$ is a double-cone formed by the intersection of the forward light-cone at $O$ and the backward light-cone at $\II$. Finally, the 3-hyperplane $\P=\{M\vert X_0(M)=1\}$ is called Bloch hyperplane, and  $\S=\M \cap \P$ is known as the Bloch ball~\cite{Bengtsson2006a,Nielsen2010a}.  

A system of two qubits AB is described by the tensor product $\H_A \otimes \H_B$ where $\H_A$ and $\H_B$ are 2D Hilbert spaces. Operators acting on $\H_A$ will be denoted by $A_i$ or labelled by super/subscript $A$ such as $\II^A$, or $\sigma_i^A$; an analogous convention is applied to B. Let $\rho$ be a density operator, or state, of the system, that is, a positive unit-trace operator on $\H_A \otimes \H_B$. Using Pauli matrices, a density operator can be written as $\rho = \frac{1}{4} \sum_{i,j=0}^{3} \Theta_{ij} (\rho) \sigma_i^A \otimes \sigma_j^B$, where $\Theta_{ij} (\rho) = \Tr [\rho (\sigma_i^A \otimes \sigma_j^B)]$.

Now suppose Alice owns part A and Bob owns part B of the system. A positive operator valued measurement (POVM) performed on A is a decomposition of the identity operator $\II^A$ into some members $\{E^A_i\}_{i=1}^{n}$ of the measurement outcomes $\M_A$, $\II^A=\sum_{i=1}^{n} E^A_i$. If Alice gets a measurement outcome $E^A_i$, the unnormalized state of Bob's system will be found to be $\rho^{B}_{i}=\Tr_A [\rho (E^A_i \otimes \II^B)]$, where $\Tr_A$ denotes the partial trace operation over subsystem A~\cite{Nielsen2010a}. Note that $\Tr(\rho_i^B)$ is the probability of observing the measurement outcome $E_i^A$. Now a key observation is that this correspondence between $E^A_i$ and $\rho^B_i$ is independent of the POVM that contains the outcome $E^A_{i}$. The correspondence establishes a map, called Alice's \emph{EPR map} -- not to be confused with the steering map as defined in~\cite{Moroder2014a}.

More precisely, the EPR map $\rho^{A \to B}$ of a state $\rho$ is a positive linear map  $\rho^{A \to B}: B^{H}(\H_A) \rightarrow B^{H}(\H_B)$ defined by $\rho^{A \to B} (A) = \Tr_A [\rho (A \otimes \II^B)]$. If $U_A$ is an element or a subset of $B^{H}(\H_A)$, we denote its image under the EPR map by $U_A'=\rho^{A \to B} (U_A)$. In particular, $\M_A'$ is called Alice's \emph{steering outcomes}. Bob's EPR map and Bob's steering outcomes are defined analogously. In fact, the defined Alice's EPR map is simply the inverse of the so-called Pillis-Jamio{\l}kowski isomorphism, which maps a linear map between two operator spaces to an operator acting on their tensor product~\cite{Pillis1967a,Jamiokowski1972a}; see also~\cite{Jiang2013a}. 

\begin{figure}[thb!]
\begin{center}
\includegraphics[width=0.48\textwidth]{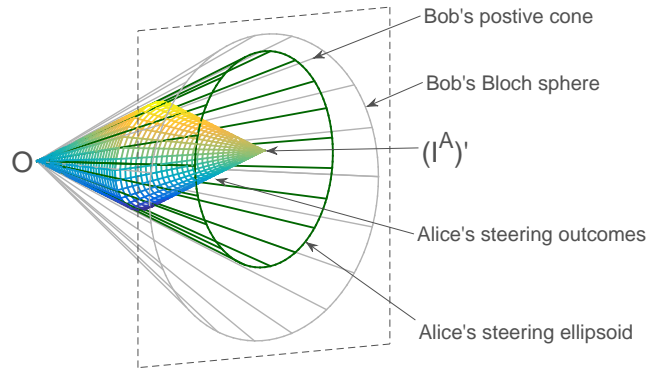}
\end{center}
\caption{(Color online) Three-dimensional representation of Alice's steering outcomes $\M_A'$ inside Bob's positive light-cone $\M_B^+$ together with their projective projections on the Bob's Bloch hyperplane.}
\label{fig: 3d-cones}
\end{figure}

If we use the Euclidean coordinates $\{X_i\}_{i=0}^{3}$ to represent the operators of $B^H(\H_A)$ and $B^H(\H_B)$, the EPR map $\rho^{A \to B}$ is simply a map from $\RR^4$ to $\RR^4$. More explicitly, it is easy to show that 
\begin{equation}
X_i (A')  = \frac{1}{2} \sum_{j=0}^{3} \Theta^{T}_{ij}(\rho) X_j (A).
\end{equation}
 
Figure~\ref{fig: 3d-cones} illustrates a 3D cross-section of Alice's steering outcomes $\M_A'$ relative to Bob's positive cone $\M_B^+$. One observes that $\M_A'$  is a skewed double-cone with two vertices at $O$ and $(\II^A)'=\Tr_A (\rho) \in \S_B$. The latter is also known as the reduced state of B~\cite{Nielsen2010a}. In the following, we show that the geometry of Alice's steering outcomes $\M_A'$ determines the non-separability and steerability of the state from Alice's side. The projective projection~\cite{Kadison1995a} of $\M_A'$ through the origin $O$ onto Bob's 3D Bloch hyperplane (see Figure~\ref{fig: 3d-cones}), known as Alice's \emph{steering ellipsoid}, has been studied in detail~\cite{Jevtic2014a,Milne2015a,Jevtic2015a}. As a result, the nested tetrahedron criterion for separability, which states that a state is separable if and only if Alice's steering ellipsoid is contained in a tetrahedron that fits in Bob's Bloch ball, has been discovered~\cite{Jevtic2014a}. However, the criterion appears to us as a mysterious fact. We will state a criterion for separability in terms of the 4D geometry of $\M_A'$ which demystifies the nested tetrahedron criterion. Moreover, this also provides a transition to studying the geometrical nature of steerability.

\subsection*{Polyhedral boxes, packing and packability}
To study the geometry of Alice's steering outcomes, we introduce the concepts of polyhedral box and packability. Let $\U_B=\{B_i\}_{i=1}^{m}$ be a set of $m$ hermitian operators acting on $\H_B$, or equivalently, vectors of $\RR^4$. As in standard convex analysis~\cite{Rockafellar1970a}, the set $\cone(\U_B)=\{B=\sum_{i=1}^{m} \alpha_i B_i \vert \alpha_i \ge 0 \}$ is called the conical hull based on $\U_B$. Further, we define the \emph{polyhedral box} based on $\U_B$ to be $\cell (\U_B)= \{\sum_{i=1}^m \beta_i B_i \vert 0 \le \beta_i \le 1\}$. Such a polyhedral box can also be seen as a linear image in the 4D space of the unit $m$-cube.  The vertex $\sum_{i=1}^m B_i$ is called the \emph{principal vertex}. The set of Alice's steering outcomes $\M_A'$ is called \emph{$m$-packable} if it is contained in a polyhedral box with the principal vertex at $(\II^A)'$. The set of steering outcomes $\M_A'$ is called \emph{packable} if it is $m$-packable for some $m$.
\subsection*{Separability}
A state $\rho$ over AB is called \emph{separable} if it can be written as a convex combination of some $s$ product states $\{\rho^A_i \otimes \rho^B_i\}_{i=1}^{s}$, i.e., $\rho= \sum_{i=1}^{s} p_i \rho^A_i \otimes \rho^B_i$, where $0 \le p_i \le 1$, $\sum_{i=1}^{s} p_i = 1$~\cite{Werner1989a}. For a two-qubit system, any separable state can be written in this form with $s \le 4$~\cite{Sanpera1998a}. The following proposition reveals a surprising connection between separability and packability.
\begin{proposition}
\label{pros:separability}
A two-qubit state $\rho$ is separable if and only if the set of Alice's steering outcomes $\M_A'$ is $4$-packable. 
\end{proposition}
If the set of Alice's steering outcomes is $4$-packable by linearly dependent operators, then Alice's steering ellipsoid is necessarily degenerate, and the state is separable~\cite{Jevtic2014a,Milne2015a}. Thus, in Proposition~\ref{pros:separability} the 4 vectors that form the polyhedral box to pack Alice's steering outcomes can be assumed to be linearly independent. In the Appendix, we show that this proposition is equivalent to the nested tetrahedron criterion for separability. The key to this equivalence is that the packability of $\M_A'$ by linearly independent positive operators is fully characterized by the cone $(\M_A^+)'$ (Appendix, Lemma~\ref{lem:independent-decomposable-sp}). The cone $(\M_A^+)'$ in turn can be characterized by its projective projection on Bob's Bloch hyperplane--the steering ellipsoid (Appendix, Lemma~\ref{lem:projective-principle-sp}). This is no longer true for packability with linearly dependent operators, in particular, for $m$-packability with $m > 4$. It then becomes clear that the limit number of $m=4$ in Proposition~\ref{pros:separability}, or the notion of tetrahedron in the nested tetrahedron criterion, appears due to the fact that it is the maximum number of linearly independent operators in $\B^H(\H_B)$. Moreover, it also suggests that steerability, which is equivalent to $m$-packability with $m$ possibly bigger than $4$ as stated in Proposition~\ref{pros:steerability}, is of fundamentally 4D geometry and cannot be seen completely in the projective projection of Alice's steering outcomes on Bob's Bloch hyperplane.

\subsection*{Steerability}
A subset $\C$ of all POVMs on a system is also called a measurement class $\C$. Relevant classes of measurements are projective measurements, where the measurement outcomes are orthogonal projections~\cite{Nielsen2010a}, and binary outcome POVMs. For a qubit, the former is a subclass of the latter.

Following~\cite{Wiseman2007a}, a state $\rho$ is called \emph{unsteerable} from Alice's side with respect to measurements of class $\C^A$ if there exists a decomposition of $(\II^A)'=\Tr_A(\rho)$ into an ensemble of $m$ positive operators of $\H_B$, $(\II^A)'= \sum_{i=1}^{m} B_i$, satisfying the following condition. For any measurement with $n$ outcomes $\{E^A_i\}_{i=1}^{n}$ of class $\C^A$ performed by Alice, the corresponding conditional ensemble of Bob's system B, $\{(E^A_i)'=\Tr_A [\rho (E^A_i \otimes \II^B)]\}_{i=1}^n$, can be expressed by a stochastic map from $\{B_i\}_{i=1}^m$ to $\{(E^A_i)'\}_{i=1}^n$, i.e.,
\begin{equation}
(E^A_i)' = \sum_{j=1}^{m} G_{ij} B_j,
\label{eq:unsteerable-def}
\end{equation}
where $G$ is a stochastic matrix, $0 \le G_{ij} \le 1$, $\sum_{i=1}^{n} G_{ij} = 1$. The set $\{B_i\}_{i=1}^m$ is called  the ensemble of \emph{local hidden states} (LHSs) for steering from Alice's side, which together with the stochastic map $G$ allow her to locally simulate steering~\cite{Wiseman2007a}. 

Determining the steerability of a state is a hard problem when one considers all possible POVMs~\cite{Werner2014a}. Most approaches are restricted to projective measurements. Nevertheless, for a system of two qubits, steerability with respect to all projective measurements is equivalent to steerability with respect to binary outcome POVMs (Appendix, Lemma~\ref{lem:projective-binary-sp}). The following proposition subsequently shows that the steerability from Alice's side with respect to binary outcome POVMs is completely determined by the geometry of her steering outcomes.

\begin{proposition}
\label{pros:steerability}
A two-qubit state $\rho$ is unsteerable from Alice's side for all binary outcome POVMs if and only if the set of Alice's steering outcomes $\M_A'$ is packable. 
\end{proposition} 
This is a non-trivial statement; generalizations to all POVMs and to higher dimensional systems are open problems with subtle difficulties. In the following, steerability will always be considered with respect to projective measurements.

Practically,  Proposition~\ref{pros:steerability} simplifies the problem of determining the steerability of two-qubit states. To find a necessary condition for steerability, one can choose some ansatz for the base $\U_B$, and check if the steering outcomes $\M_A'$ stay within the polyhedral box based on $\U_B$, in which case the state is unsteerable from Alice's side. Although an ansatz for $\U_B$ can also be considered as an ansatz for the ensemble of LHSs for steering from Alice's side, this approach shows that a given ansatz naturally generates a necessary condition for steerability for any state with B's reduced state at the principal vertex of $\cell(\U_B)$; in that sense, an ansatz for the ensemble of LHSs can be fully exploited. The main task in this procedure is to determine the boundary of the polyhedral box for a given ansatz $\U_B$, which is the subject of the following subsection. 

\subsection*{Determining the boundary of polyhedral boxes}
An ansatz $\U_B$ can always be chosen such that its vectors are on the boundary of $\M_B^+$. Indeed, if a vector of $\U_B$ is not on the boundary of $\M_B^{+}$, one can decompose it into a sum of two vectors on the boundary of $\M_B^{+}$ to form a new ansatz, whose polyhedral box contains the polyhedral box of the old ansatz. Each vector of such an ansatz is of the form $u_i \begin{pmatrix} 1 \\ \vec{n_i} \end{pmatrix}$, where $u_i$ determines its length and 
$\vec{n_i}$ is a 3D unit vector. More generally, each ansatz can be characterized by a distribution $u(\n)$ on the 
3D unit sphere. 

Any vector of $\cell(\U_B)$ is of the form $\int \d \mu (\n) f(\n)  \begin{pmatrix} 1 \\ \n \end{pmatrix}$, where $0 \le f(\n) \le 1$ and $\d \mu( \n)$ denotes the measure on the 3D unit sphere generated by the distribution $u$. In particular, the principal vertex is $\int \d \mu ( \n) \begin{pmatrix} 1 \\ \n \end{pmatrix}$. This principal vertex should be on Bob's Bloch hyperplane, thus one has a normalization condition $\int \d \mu (\n)=1$. 

The cross-section of $\cell (\U_B)$ at some hyperplane $X_0 =x_0$, denoted $\cell(\U_B) \vert_{x_0}$, consists of vectors $\vec{b}= \int \d \mu ( \n)  f(\n) \n$ with $0 \le f(\n) \le 1$ and $x_0 = \int \d \mu (\n) f(\n)$. To determine the boundary of $\cell(\U_B) \vert_{x_0}$, we choose a direction $\n_0$ on the 3D unit sphere and project it onto that direction. The extreme point of the projection of the convex set $\cell(\U_B) \vert_{x_0}$  must be the projection of a point on its boundary. We are thus led to solving the optimization problem
\begin{equation}
\max_{0 \le f(\n) \le 1} \int \d \mu(\n) f(\n) \n_0^T \n  
\end{equation}
with constraint $x_0 = \int \d \mu (\n) f(\n)$. Using the method of Lagrange multipliers, we consider the modified objective function
\begin{equation}
L(f, \lambda) = \int  d\mu (\n ) \left [ f(\n) (\n_0^T \n - \lambda) + \lambda x_0 \right ].
\end{equation} 
When $\lambda$ is fixed, $L$ obtains its extremal value at the functions $f$ of the form $f_{\n_0} (\n) = 1_{\n_0^T \n > \lambda} (\n) + g(\n) 1_{\n_0^T \n = \lambda} (\n)$, where $g$ is any function taking values in $[0,1]$ and $1_X$ denotes the indicator function of a set $X$, $1_X(\n)=1$ if $\n \in X$ and $1_X(\n)=0$ if $\n \not\in X$. Each solution $f_{\n_0}(\n)$ then gives a point on the boundary of $\cell(\U_B)$ parametrized by $\lambda, \n_0$ and $g$,


\begin{align}
x_0 &= \int \d \mu(\n) \left [ 1_{\n_0^T \n > \lambda} (\n) + g(\n) 1_{\n_0^T \n =\lambda} (\n) \right  ],\\
\vec{b}&= \int \d \mu(\n) \left  [1_{\n_0^T \n > \lambda}(\n) + g(\n) 1_{\n_0^T \n=\lambda} (\n) \right ] \n.
\label{eq:boundary-map}
\end{align}
In the case $\mu$ is sufficiently fine (e.g., $u$ is continuous), the latter terms are integrals over zero-measure sets, thus vanish, and $g$ is irrelevant. In the other cases (e.g, $u$ has $\delta$-peaks), $\n_0^T \n=\lambda$ may be of non-zero measure. In fact, in these cases, $\cell(\U_B) \vert_{x_0}$ may have degenerate flat regions and $g$ allows one to get all the points on these flat regions.

The simplest case where these integrals can be calculated explicitly is when $u(\n)$ is a uniform distribution. In this case, $\mu$ is fine and $g$ is irrelevant. In fact, the cross-section of  $\cell{(\U_B)}$ at $X_0=x_0$ is a ball of radius $r_{\mathrm{uni.}}(x_0) = x_0 (1-x_0)$. Since in this case the principal vertex of $\cell(\U_B)$ is at the center of Bob's Bloch ball, this ansatz can be used to find a necessary condition for steerability for any state that has B's reduced state completely mixed. 

\subsection*{Example: Werner states and their modification}
Werner states are defined by
\begin{equation}
W^p= p \ketbra{\Phi^{+}}{\Phi^{+}} + (1-p) \frac{\II^A}{2} \otimes \frac{\II^B}{2},
\end{equation}
where $\ket{\Phi^{+}}=\frac{1}{\sqrt{2}}(\ket{0,0}+\ket{1,1})$ is one of the Bell states and $0 \le p \le 1$~\cite{Werner1989a}. Using $\Theta_{ij} (W^p)=\Tr [W^p (\sigma_i^A \otimes \sigma_j^B)]$, one finds the matrix presentation of Alice's EPR map, $\frac{1}{2}\Theta^T (W^p)= \frac{1}{2}\operatorname{diag}(1,p,-p,p)$. The EPR map contracts the $X_0$-axis a factor of $\frac{1}{2}$, and every other axis a factor of $\frac{p}{2}$. Although the $X_2$-axis is also reflected, this is irrelevant to the geometry of $\M_A'$. 
When $p=1$, the Werner state is pure and $\M_A'$ touches the boundary of $\M_B^+$. On the other hand, when $p=0$, the Werner state is the product of two completely mixed states and $\M_A'$ shrinks to a single line-segment.

\begin{figure}[t!]
\begin{center}
\includegraphics[width=0.49\textwidth]{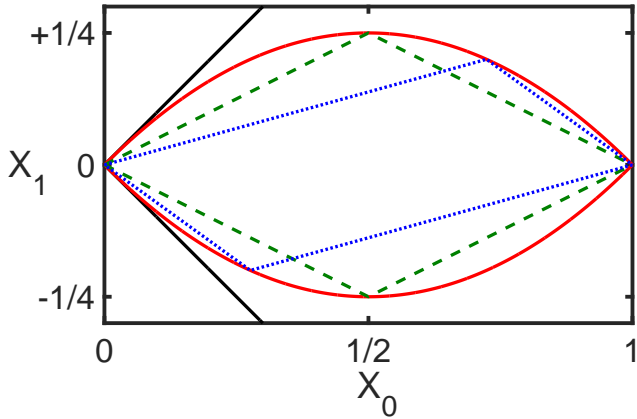}
\end{center}
\caption{(Color online) Two-dimensional cross-sections of the boundaries of the polyhedral box based on the uniform ansatz (red, solid), the steering outcomes of the Werner state at $p=\frac{1}{2}$ (green, dashed), and the steering outcomes of the modified Werner state at $p = 0.4$, $q \approx 0.75$ (blue, dotted). The outer-most black lines present the cross-section of the forward light-cone at the origin.}
\label{fig:limline}
\end{figure}

Elementary geometry shows that the set of steering outcomes $\M_A'$ is $4$-packable, or the Werner state is separable, if and only if $p \le \frac{1}{3}$. To find a sufficient condition for the Werner state to be unsteerable we use the uniform ansatz for $\U_B$. In fact, due to the spherical symmetry of Werner states, it is easy to see that this condition is also necessary~\cite{Wiseman2007a}. The boundary of $\cell (\U_B)$ is illustrated in Figure~\ref{fig:limline} together with the boundary of $\M_A'$. One finds that for $\M_A'$ to stay within this boundary, one needs $\frac{p}{2} \le r_{\mathrm{uni.}}(\frac{1}{2})$, or $p \le \frac{1}{2}$. We thus recovered the well-known results for Werner states regarding their non-separability and steerability~\cite{Werner1989a,Wiseman2007a}.

An advantage of using the 4D geometrical description is that the boundary of $\cell (\U_B)$, once determined, can be used to find a necessary condition for other states to be steerable as well. As an example, we consider the following modified Werner states,
\begin{equation}
\tilde{W}^p_{q}= p \ketbra{\Phi^{+}}{\Phi^{+}} + (1-p) \frac{\II^A + q \sigma_z^A}{2}\otimes \frac{\II^B}{2} ,
\end{equation}
with $0 \le \abs{q} \le 1$.  One notices that $\Tr_A (\tilde{W}_p^q)= \frac{\II^B}{2}$, thus the uniform ansatz is valid. The matrix for the EPR map is $\frac{1}{2} \Theta^{T} (\tilde{W}^p_{q}) = \frac{1}{2}[\operatorname{diag} (1,p,-p,p) + q (1-p) \delta_{1,4} ]$, where $\delta_{i,j}$ is the Kronecker matrix. The boundary of the steering outcomes $\M_A'$ for this state is also illustrated in Figure~\ref{fig:limline}. One easily finds that the condition for $\M_A'$ to be contained in $\cell(\U_B)$, which implies the unsteerability of the Werner state, is $\frac{\sqrt{1- 2 p}}{1-p} \ge \abs{q}$. Although this inequality can also be deduced from a recent result of Bowles~\textit{et al.}~\cite{Bowles2016a}, we have arrived at it simply based on the geometry of steering outcomes. 

\section*{Conclusion}
By defining EPR maps, we are able to map the properties of a joint state of two qubits, namely non-separability and steerability, to the geometrical properties of steering outcomes. On one hand, our analysis clarifies the nested tetrahedron criterion for separability. On the other hand, we establish a general framework to determine the necessary condition for steerability. That this framework allows one to show the optimality of a LHS model will be discussed in a subsequent work~\cite{Nguyen2016b}. Our work further opens new interesting questions. Although steerability with binary POVMs is a geometrical property, it remains to be clarified if this is still true for general POVMs. The question whether Bell non-locality can be reduced to the geometry of steering outcomes is also to be explored. 

\begin{acknowledgments}
We would like to thank Xuan Thanh Le and Kimmo Luoma for helpful discussions. Comments from Alan Celestino, Anna Deluca, Michael Hall, Sania Jevtic, Antony Milne, Ba An Nguyen and Huy Viet Nguyen on our early manuscript are gratefully acknowledged.
\end{acknowledgments}

\appendix
\section*{Appendix}
\setcounter{lemma}{0}
\setcounter{proposition}{0}

In this Appendix, we provide details for the claims stated in the main text.

\subsection*{Separability}
In this subsection, we give the proof of Proposition 1 saying that a state is separable if and only if the Alice's steering outcomes is $4$-packable.

\begin{proposition}
\label{pros:separability-sp}
A two-qubit state $\rho$ is separable if and only if the set of Alice's steering outcomes $\M_A'$ is $4$-packable.
\end{proposition}
As reasoning in the main text, we can assume that Alice's steering outcomes is $4$-packable with 4 linearly independent vectors. The proposition then follows directly from the following lemmas.

\begin{lemma}
\label{lem:independent-decomposable-sp}
The set of steering outcomes $\M_A'$ is $4$-packable with 4 linearly independent positive operators if and only if the cone $(\M_A^{+})'$ is contained in a conical hull of 4 linearly independent positive operators. 
\end{lemma}
\begin{proof} 
Assume that $\M_A'$ is $4$-packable by a set of $4$ linearly independent vectors $\U_B= \{B_i \vert 0 \le i \le 4\}$. For each $i$, denote $H_i$ the hyperplane spanned by $\U_B \setminus \{B_i\}$. Denote $H_i^+$ the half-space divided by $H_i$ that contains $B_i$. Since $\M_A' \subseteq \cell(\U_B)$, it is also contained in $H_i^+$. This implies that $(\M_A^+)' = \cone (\M_A') \subseteq \cone (H_i^+) = H_i^+$ for every $i$, or  $(\M_A^{+})' \subseteq \cap_{i=1}^{4} H_i^{+}$. On the other hand, for linearly independent operators $B_i$, one has $\cap_{i=1}^{4}  H_i^{+}= \cone (\U_B)$, thus  $(\M_A^{+})' \subseteq \cone (\U_B)$.

For the reverse direction, we assume that the cone $(\M_A^+)'$ is contained in a conical hull, $\cone(\U_B)$, formed by $4$ linearly independent positive operators $B_i, 1\le i \le 4$. Note that $\frac{1}{2}(\II^A)'$ is the center of symmetry of $\M_A'$. Denote $\cone(\U_B)^-$ the reflection of $\cone(\U_B)$ through $\frac{1}{2}(\II^A)'$. Since $(\M_A^-)'$ is the reflection of $(\M_A^+)'$ through $\frac{1}{2}(\II^A)'$, we deduce that $(\M_A^-)' \subseteq \cone(\U_B)^-$. Thus $\M_A' = (\M_A^+)' \cap (\M_A^-)'$ is contained in $\cone(\U_B) \cap \cone(\U_B)^-$. When $B_i$ are linearly independent, the latter is a polyhedral box based on $\tilde{\U}_B=\{\tilde{B}_i\}_{i=1}^{4}$, where $\tilde{B}_i= \lambda_{i} B_i $ with $\lambda_i=\max \{\lambda \vert \lambda B_i \in \cone(\U_B)^{-} \}$.
\end{proof}

\begin{lemma} 
\label{lem:projective-principle-sp}
The cone $(\M_A^{+})'$ is contained in a conical hull formed by 4 linearly independent positive operators if and only if  $\S_A'$ is contained in a tetrahedron which is contained in Bob's Bloch ball $\S_B$.
\end{lemma}
\begin{proof}
Assume that $(\M_A^+)'$ is contained in the conical hull based on $\U_B= \{B_i \vert 1 \le i \le 4 \}$. Let $\P_B$ denote Bob's Bloch hyperplane $X_0(M) = 1$, then $(\M_A^+)' \cap \P_B = \S_A' \subseteq \cone (\U_B) \cap \P_B$. Since $B_i$ are linearly independent, $\cone(\U_B) \cap \P_B$ is a tetrahedron. Moreover, this tetrahedron is contained in Bob's Bloch ball $\S_B$ because $\U_B$ consists of positive operators.

For the reverse direction, assume that $\S_A'$ is contained in a tetrahedron which is the convex hull of $\U_B = \{B_i \vert 1 \le i \le 4 \}$. The operators $B_i$ are linearly independent, otherwise the tetrahedron is degenerate. Now, we have $(\M_A^{+})' = \cone (\S_A') \subseteq \cone (\U_B)$, which is the required condition.
\end{proof}

\subsection*{Steerability}
We first show that for a two-qubit state, steerablity with respect to binary outcome POVMs is equivalent to steerability with respect to all projective measurements.

\label{sec:steerability}
\begin{lemma}
\label{lem:projective-binary-sp}
A two-qubit state $\rho$ is unsteerable from Alice's side with respect to all projective measurements if and only if it is unsteerable from Alice's side with respect to all binary outcome POVMs.
\end{lemma}
\begin{proof}
It is obvious that unsteerability with respect to binary outcome POVMs implies unsteerability with projective measurements. We prove the converse statement. Suppose a state $\rho$ is unsteerable with respect to all projective measurements from Alice's side and $\U_B=\{B_i\}_{i=1}^{m}$ is an ensemble of LHSs for steering from Alice's side. Suppose Alice makes a binary POVM $\II^A= E^A_1+E^A_2$. Consider the spectral decomposition of $E^A_1$, $E^A_1= H_{11} P^A_1+ H_{12} P^A_2$, where $P^A_1$ and $P^A_2$ form a complete set of two orthogonal projections, and $H_{11}$ and $H_{12}$ are positive eigenvalues of $E^A_1$. Then the spectral decomposition of $E^A_2$ is $E^A_2= (1-H_{11}) P^A_1 + (1-H_{12}) P^A_2$. That is to say, $E^A_i= \sum_{j=1}^2 H_{ij} P^A_j$, where $H$ is a $(2 \times 2)$ stochastic matrix with $H_{21}=1-H_{11}$, $H_{22}=1-H_{12}$. Since $P^A_1$ and $P^A_2$ constitute a projective measurement by Alice, and since $\U_B$ is an ensemble of LHSs for steering from Alice's side with projective measurements, it follows that there exists a $(2 \times m)$ stochastic matrix $K$ such that $(P^A_i)'=\sum_{j=1}^{m} K_{ij} B_j$. But this also implies that $(E^A_i)'= \sum_{j=1}^{2} \sum_{k=1}^{m} H_{ij} K_{jk} B_k$. Since $H$ and $K$ are stochastic matrices, $G=HK$ is also a stochastic matrix. Therefore the state is also unsteerable with respect to all binary outcome measurements.
\end{proof}
 
Furthermore, we show:
\begin{proposition}
\label{pros:steerability-sp}
A two-qubit state $\rho$ is unsteerable from Alice's side for all binary outcome POVMs if and only if the set of Alice's steering outcomes $\M_A'$ is packable. 
\end{proposition} 
\begin{proof}
Suppose for binary outcome measurements, the state is unsteerable from Alice's side. Then there exists a set $\U_B=\{B_i\}_{i=1}^{m}$ of $m$ positive operators playing the role of an ensemble of LHSs for steering from Alice's side. We will show that $\M_A' \subseteq \cell(\U_B)$. Indeed, take $B \in \M_A'$, then there exists an operator $E^A_1 \in \M_A$ such that $(E^A_1)' = B$. Let $E^A_2= \II^A-E^A_1$, then $\{E^A_1,E^A_2\}$ constitutes a binary outcome POVM performed by Alice. By definition of an ensemble of LHSs, there exists a $(2 \times m)$ stochastic matrix $G$ such that $(E^A_i)' = \sum_{j=1}^{m} G_{ij} B_j$, in particular $B=(E^A_1)'= \sum_{j=1}^{m} G_{1j} B_j$. Since $0 \le G_{1j} \le 1$, this implies that $B \in \cell(\U_B)$. 

Now suppose  $\M_A' \subseteq \cell(\U_B)$, with $\U_B=\{B_i\}_{i=1}^{m}$, $(\II^A)' = \sum_{i=1}^{m} B_i$, we show that $\U_B$ can play the role of an ensemble of LHSs for all binary measurements performed by Alice. Any binary outcome POVM performed by Alice is of the form $\II^A=E^A_1+E^A_2$ with $E^A_i \in \M_A$. This POVM induces a decomposition of the reduced state of B,  $(\II^A)'= (E^A_1)' + (E^A_2)'$. Because $(E^A_1)' \in \M_A' \subseteq \cell(\U_B)$, there exist $m$ numbers $\{0 \le G_{1j} \le 1\}_{j=1}^{m}$ such that $(E^A_1)'=\sum_{j=1}^{m} G_{1j} U^B_j$. It follows that $(E^A_2)'=(\II^A)'-(E^A_1)'=\sum_{j=1}^{m} [1-G_{1j}] U^B_j$. Let $G_{2j}=1-G_{1j}$, then $G$ is a $(2 \times m)$ stochastic matrix that satisfies $(E^A_i)'=\sum_{j=1}^{m} G_{ij} U^B_j$ for $i=1$, $2$. This implies that $\U_B$ is an ensemble of LHSs and the state is unsteerable from Alice's side.
\end{proof}

\bibliography{../bibtex/quantum-steering.bib}

\begin{thebibliography}{35}%
\makeatletter
\providecommand \@ifxundefined [1]{%
 \@ifx{#1\undefined}
}%
\providecommand \@ifnum [1]{%
 \ifnum #1\expandafter \@firstoftwo
 \else \expandafter \@secondoftwo
 \fi
}%
\providecommand \@ifx [1]{%
 \ifx #1\expandafter \@firstoftwo
 \else \expandafter \@secondoftwo
 \fi
}%
\providecommand \natexlab [1]{#1}%
\providecommand \enquote  [1]{``#1''}%
\providecommand \bibnamefont  [1]{#1}%
\providecommand \bibfnamefont [1]{#1}%
\providecommand \citenamefont [1]{#1}%
\providecommand \href@noop [0]{\@secondoftwo}%
\providecommand \href [0]{\begingroup \@sanitize@url \@href}%
\providecommand \@href[1]{\@@startlink{#1}\@@href}%
\providecommand \@@href[1]{\endgroup#1\@@endlink}%
\providecommand \@sanitize@url [0]{\catcode `\\12\catcode `\$12\catcode
  `\&12\catcode `\#12\catcode `\^12\catcode `\_12\catcode `\%12\relax}%
\providecommand \@@startlink[1]{}%
\providecommand \@@endlink[0]{}%
\providecommand \url  [0]{\begingroup\@sanitize@url \@url }%
\providecommand \@url [1]{\endgroup\@href {#1}{\urlprefix }}%
\providecommand \urlprefix  [0]{URL }%
\providecommand \Eprint [0]{\href }%
\providecommand \doibase [0]{http://dx.doi.org/}%
\providecommand \selectlanguage [0]{\@gobble}%
\providecommand \bibinfo  [0]{\@secondoftwo}%
\providecommand \bibfield  [0]{\@secondoftwo}%
\providecommand \translation [1]{[#1]}%
\providecommand \BibitemOpen [0]{}%
\providecommand \bibitemStop [0]{}%
\providecommand \bibitemNoStop [0]{.\EOS\space}%
\providecommand \EOS [0]{\spacefactor3000\relax}%
\providecommand \BibitemShut  [1]{\csname bibitem#1\endcsname}%
\let\auto@bib@innerbib\@empty
\bibitem [{\citenamefont {Einstein}\ \emph {et~al.}(1935)\citenamefont
  {Einstein}, \citenamefont {Podolsky},\ and\ \citenamefont
  {Rosen}}]{Einstein1935a}%
  \BibitemOpen
  \bibfield  {author} {\bibinfo {author} {\bibfnamefont {A.}~\bibnamefont
  {Einstein}}, \bibinfo {author} {\bibfnamefont {B.}~\bibnamefont {Podolsky}},
  \ and\ \bibinfo {author} {\bibfnamefont {N.}~\bibnamefont {Rosen}},\
  }\bibfield  {title} {\enquote {\bibinfo {title} {Can quantum-mechanical
  description of physical reality be considered complete},}\ }\href@noop {}
  {\bibfield  {journal} {\bibinfo  {journal} {Phys. Rev.}\ }\textbf {\bibinfo
  {volume} {47}},\ \bibinfo {pages} {777} (\bibinfo {year} {1935})}\BibitemShut
  {NoStop}%
\bibitem [{\citenamefont {Schr\"{o}dinger}(1935)}]{Schrodinger1935a}%
  \BibitemOpen
  \bibfield  {author} {\bibinfo {author} {\bibfnamefont {E.}~\bibnamefont
  {Schr\"{o}dinger}},\ }\bibfield  {title} {\enquote {\bibinfo {title}
  {Discussion of probability relations between separated systems},}\
  }\href@noop {} {\bibfield  {journal} {\bibinfo  {journal} {Proc. Cambridge
  Philos. Soc.}\ }\textbf {\bibinfo {volume} {31}},\ \bibinfo {pages} {555}
  (\bibinfo {year} {1935})}\BibitemShut {NoStop}%
\bibitem [{\citenamefont {Bell}(1964)}]{Bell1964a}%
  \BibitemOpen
  \bibfield  {author} {\bibinfo {author} {\bibfnamefont {J.~S.}\ \bibnamefont
  {Bell}},\ }\bibfield  {title} {\enquote {\bibinfo {title} {On the
  {Einstein-Podolsky-Rosen} paradox},}\ }\href@noop {} {\bibfield  {journal}
  {\bibinfo  {journal} {Physics}\ }\textbf {\bibinfo {volume} {1}},\ \bibinfo
  {pages} {195} (\bibinfo {year} {1964})}\BibitemShut {NoStop}%
\bibitem [{\citenamefont {Werner}(1989)}]{Werner1989a}%
  \BibitemOpen
  \bibfield  {author} {\bibinfo {author} {\bibfnamefont {R.~F.}\ \bibnamefont
  {Werner}},\ }\bibfield  {title} {\enquote {\bibinfo {title} {Quantum states
  with {Einstein-Podolsky-Rosen} correlations admitting a hidden-variable
  model},}\ }\href@noop {} {\bibfield  {journal} {\bibinfo  {journal} {Phys.
  Rev. A}\ }\textbf {\bibinfo {volume} {40}},\ \bibinfo {pages} {4277}
  (\bibinfo {year} {1989})}\BibitemShut {NoStop}%
\bibitem [{\citenamefont {Wiseman}\ \emph {et~al.}(2007)\citenamefont
  {Wiseman}, \citenamefont {Jones},\ and\ \citenamefont
  {Doherty}}]{Wiseman2007a}%
  \BibitemOpen
  \bibfield  {author} {\bibinfo {author} {\bibfnamefont {H.~M.}\ \bibnamefont
  {Wiseman}}, \bibinfo {author} {\bibfnamefont {S.~J.}\ \bibnamefont {Jones}},
  \ and\ \bibinfo {author} {\bibfnamefont {A.~C.}\ \bibnamefont {Doherty}},\
  }\bibfield  {title} {\enquote {\bibinfo {title} {Steering, entanglement,
  nonlocality, and the {Einstein-Podolsky-Rosen} paradox},}\ }\href@noop {}
  {\bibfield  {journal} {\bibinfo  {journal} {Phys. Rev. Lett.}\ }\textbf
  {\bibinfo {volume} {98}},\ \bibinfo {pages} {140402} (\bibinfo {year}
  {2007})}\BibitemShut {NoStop}%
\bibitem [{\citenamefont {Jones}\ \emph {et~al.}(2007)\citenamefont {Jones},
  \citenamefont {Wiseman},\ and\ \citenamefont {Doherty}}]{Jones2007b}%
  \BibitemOpen
  \bibfield  {author} {\bibinfo {author} {\bibfnamefont {S.~J.}\ \bibnamefont
  {Jones}}, \bibinfo {author} {\bibfnamefont {H.~M.}\ \bibnamefont {Wiseman}},
  \ and\ \bibinfo {author} {\bibfnamefont {A.~C.}\ \bibnamefont {Doherty}},\
  }\bibfield  {title} {\enquote {\bibinfo {title} {Entanglement,
  {Einstein-Podolsky-Rosen} correlations, bell nonlocality, and steering},}\
  }\href@noop {} {\bibfield  {journal} {\bibinfo  {journal} {Phys. Rev. A}\
  }\textbf {\bibinfo {volume} {76}},\ \bibinfo {pages} {052116} (\bibinfo
  {year} {2007})}\BibitemShut {NoStop}%
\bibitem [{\citenamefont {Popescu}\ and\ \citenamefont
  {Rohrlich}(1994)}]{Popescu1994a}%
  \BibitemOpen
  \bibfield  {author} {\bibinfo {author} {\bibfnamefont {S.}~\bibnamefont
  {Popescu}}\ and\ \bibinfo {author} {\bibfnamefont {D.}~\bibnamefont
  {Rohrlich}},\ }\bibfield  {title} {\enquote {\bibinfo {title} {Causality and
  nonlocality as axioms for quantum mechanics},}\ }\href@noop {} {\bibfield
  {journal} {\bibinfo  {journal} {Found. Phys.}\ }\textbf {\bibinfo {volume}
  {24}},\ \bibinfo {pages} {379} (\bibinfo {year} {1994})}\BibitemShut
  {NoStop}%
\bibitem [{\citenamefont {Paw{\l}owski}\ \emph {et~al.}(2009)\citenamefont
  {Paw{\l}owski}, \citenamefont {Paterek}, \citenamefont {Kaszlikowski},
  \citenamefont {Scarani}, \citenamefont {Winter},\ and\ \citenamefont
  {\.{Z}ukowski}}]{Pawlowski2009a}%
  \BibitemOpen
  \bibfield  {author} {\bibinfo {author} {\bibfnamefont {M.}~\bibnamefont
  {Paw{\l}owski}}, \bibinfo {author} {\bibfnamefont {T.}~\bibnamefont
  {Paterek}}, \bibinfo {author} {\bibfnamefont {D.}~\bibnamefont
  {Kaszlikowski}}, \bibinfo {author} {\bibfnamefont {V.}~\bibnamefont
  {Scarani}}, \bibinfo {author} {\bibfnamefont {A.}~\bibnamefont {Winter}}, \
  and\ \bibinfo {author} {\bibfnamefont {M.}~\bibnamefont {\.{Z}ukowski}},\
  }\bibfield  {title} {\enquote {\bibinfo {title} {Information causality as a
  physical principle},}\ }\href@noop {} {\bibfield  {journal} {\bibinfo
  {journal} {Nature}\ }\textbf {\bibinfo {volume} {461}},\ \bibinfo {pages}
  {1101} (\bibinfo {year} {2009})}\BibitemShut {NoStop}%
\bibitem [{\citenamefont {Cavalcanti}\ \emph {et~al.}(2009)\citenamefont
  {Cavalcanti}, \citenamefont {Jones}, \citenamefont {Wiseman},\ and\
  \citenamefont {Reid}}]{Cavalcanti2009a}%
  \BibitemOpen
  \bibfield  {author} {\bibinfo {author} {\bibfnamefont {E.~G.}\ \bibnamefont
  {Cavalcanti}}, \bibinfo {author} {\bibfnamefont {S.~J.}\ \bibnamefont
  {Jones}}, \bibinfo {author} {\bibfnamefont {H.~M.}\ \bibnamefont {Wiseman}},
  \ and\ \bibinfo {author} {\bibfnamefont {M.~D.}\ \bibnamefont {Reid}},\
  }\bibfield  {title} {\enquote {\bibinfo {title} {Experimental criteria for
  steering and the {Einstein-Podolsky-Rosen} paradox},}\ }\href@noop {}
  {\bibfield  {journal} {\bibinfo  {journal} {Phys. Rev. A}\ }\textbf {\bibinfo
  {volume} {80}},\ \bibinfo {pages} {032112} (\bibinfo {year}
  {2009})}\BibitemShut {NoStop}%
\bibitem [{\citenamefont {Saunders}\ \emph {et~al.}(2011)\citenamefont
  {Saunders}, \citenamefont {Jones}, \citenamefont {Wiseman},\ and\
  \citenamefont {Pryde}}]{Saunders2011a}%
  \BibitemOpen
  \bibfield  {author} {\bibinfo {author} {\bibfnamefont {D.~J.}\ \bibnamefont
  {Saunders}}, \bibinfo {author} {\bibfnamefont {S.~J.}\ \bibnamefont {Jones}},
  \bibinfo {author} {\bibfnamefont {H.~M.}\ \bibnamefont {Wiseman}}, \ and\
  \bibinfo {author} {\bibfnamefont {G.~J.}\ \bibnamefont {Pryde}},\ }\bibfield
  {title} {\enquote {\bibinfo {title} {Experimental {EPR}-steering using
  {Bell-local} states},}\ }\href@noop {} {\bibfield  {journal} {\bibinfo
  {journal} {Nat. Phys.}\ }\textbf {\bibinfo {volume} {6}},\ \bibinfo {pages}
  {845} (\bibinfo {year} {2011})}\BibitemShut {NoStop}%
\bibitem [{\citenamefont {\.{Z}ukowski}\ \emph {et~al.}(2015)\citenamefont
  {\.{Z}ukowski}, \citenamefont {Dutta},\ and\ \citenamefont
  {Yin}}]{Zukowski2015a}%
  \BibitemOpen
  \bibfield  {author} {\bibinfo {author} {\bibfnamefont {M.}~\bibnamefont
  {\.{Z}ukowski}}, \bibinfo {author} {\bibfnamefont {A.}~\bibnamefont {Dutta}},
  \ and\ \bibinfo {author} {\bibfnamefont {Z.}~\bibnamefont {Yin}},\ }\bibfield
   {title} {\enquote {\bibinfo {title} {Geometric {Bell-like} inequalities for
  steering},}\ }\href@noop {} {\bibfield  {journal} {\bibinfo  {journal} {Phys.
  Rev. A}\ }\textbf {\bibinfo {volume} {91}},\ \bibinfo {pages} {032107}
  (\bibinfo {year} {2015})}\BibitemShut {NoStop}%
\bibitem [{\citenamefont {Marciniak}\ \emph {et~al.}(2015)\citenamefont
  {Marciniak}, \citenamefont {Rutkowski}, \citenamefont {Yin}, \citenamefont
  {Horodecki},\ and\ \citenamefont {Horodecki}}]{Marciniak2015a}%
  \BibitemOpen
  \bibfield  {author} {\bibinfo {author} {\bibfnamefont {M.}~\bibnamefont
  {Marciniak}}, \bibinfo {author} {\bibfnamefont {A.}~\bibnamefont
  {Rutkowski}}, \bibinfo {author} {\bibfnamefont {Z.}~\bibnamefont {Yin}},
  \bibinfo {author} {\bibfnamefont {M.}~\bibnamefont {Horodecki}}, \ and\
  \bibinfo {author} {\bibfnamefont {R.}~\bibnamefont {Horodecki}},\ }\bibfield
  {title} {\enquote {\bibinfo {title} {Unbounded violation of quantum steering
  inequalities},}\ }\href@noop {} {\bibfield  {journal} {\bibinfo  {journal}
  {Phys. Rev. Lett.}\ }\textbf {\bibinfo {volume} {115}},\ \bibinfo {pages}
  {170401} (\bibinfo {year} {2015})}\BibitemShut {NoStop}%
\bibitem [{\citenamefont {Kogias}\ \emph {et~al.}(2015)\citenamefont {Kogias},
  \citenamefont {Skrzypczyk}, \citenamefont {Cavalcanti}, \citenamefont
  {Ac\'{i}n},\ and\ \citenamefont {Adesso}}]{Kogia2015a}%
  \BibitemOpen
  \bibfield  {author} {\bibinfo {author} {\bibfnamefont {I.}~\bibnamefont
  {Kogias}}, \bibinfo {author} {\bibfnamefont {P.}~\bibnamefont {Skrzypczyk}},
  \bibinfo {author} {\bibfnamefont {D.}~\bibnamefont {Cavalcanti}}, \bibinfo
  {author} {\bibfnamefont {A.}~\bibnamefont {Ac\'{i}n}}, \ and\ \bibinfo
  {author} {\bibfnamefont {G.}~\bibnamefont {Adesso}},\ }\bibfield  {title}
  {\enquote {\bibinfo {title} {Hierarchy of steering criteria based on moments
  for all bipartite quantum systems},}\ }\href@noop {} {\bibfield  {journal}
  {\bibinfo  {journal} {Phys. Rev. Lett.}\ }\textbf {\bibinfo {volume} {115}},\
  \bibinfo {pages} {210401} (\bibinfo {year} {2015})}\BibitemShut {NoStop}%
\bibitem [{\citenamefont {Zhu}\ \emph {et~al.}(2016)\citenamefont {Zhu},
  \citenamefont {Hayashi},\ and\ \citenamefont {Chen}}]{Zhu2015a}%
  \BibitemOpen
  \bibfield  {author} {\bibinfo {author} {\bibfnamefont {H.}~\bibnamefont
  {Zhu}}, \bibinfo {author} {\bibfnamefont {M.}~\bibnamefont {Hayashi}}, \ and\
  \bibinfo {author} {\bibfnamefont {L.}~\bibnamefont {Chen}},\ }\bibfield
  {title} {\enquote {\bibinfo {title} {Universal steering inequalities},}\
  }\href@noop {} {\bibfield  {journal} {\bibinfo  {journal} {Phys. Rev. Lett.}\
  }\textbf {\bibinfo {volume} {116}},\ \bibinfo {pages} {070403} (\bibinfo
  {year} {2016})}\BibitemShut {NoStop}%
\bibitem [{\citenamefont {Uola}\ \emph {et~al.}(2014)\citenamefont {Uola},
  \citenamefont {Moroder},\ and\ \citenamefont {G{\"{u}}hne}}]{Uola2014a}%
  \BibitemOpen
  \bibfield  {author} {\bibinfo {author} {\bibfnamefont {R.}~\bibnamefont
  {Uola}}, \bibinfo {author} {\bibfnamefont {T.}~\bibnamefont {Moroder}}, \
  and\ \bibinfo {author} {\bibfnamefont {O.}~\bibnamefont {G{\"{u}}hne}},\
  }\bibfield  {title} {\enquote {\bibinfo {title} {Joint measurability of
  generalized measurements implies classicality},}\ }\href@noop {} {\bibfield
  {journal} {\bibinfo  {journal} {Phys. Rev. Lett.}\ }\textbf {\bibinfo
  {volume} {113}},\ \bibinfo {pages} {160403} (\bibinfo {year}
  {2014})}\BibitemShut {NoStop}%
\bibitem [{\citenamefont {Quintino}\ \emph {et~al.}(2014)\citenamefont
  {Quintino}, \citenamefont {V\'{e}rtesi},\ and\ \citenamefont
  {Brunner}}]{Quintino2014a}%
  \BibitemOpen
  \bibfield  {author} {\bibinfo {author} {\bibfnamefont {M.~T.}\ \bibnamefont
  {Quintino}}, \bibinfo {author} {\bibfnamefont {T.}~\bibnamefont
  {V\'{e}rtesi}}, \ and\ \bibinfo {author} {\bibfnamefont {N.}~\bibnamefont
  {Brunner}},\ }\bibfield  {title} {\enquote {\bibinfo {title} {Joint
  measurability, {Einstein-Podolsky-Rosen} steering, and {Bell} nonlocality},}\
  }\href@noop {} {\bibfield  {journal} {\bibinfo  {journal} {Phys. Rev. Lett.}\
  }\textbf {\bibinfo {volume} {113}},\ \bibinfo {pages} {160402} (\bibinfo
  {year} {2014})}\BibitemShut {NoStop}%
\bibitem [{\citenamefont {Chen}\ \emph {et~al.}(2014)\citenamefont {Chen},
  \citenamefont {Li}, \citenamefont {Lambert}, \citenamefont {Chen},
  \citenamefont {Ota}, \citenamefont {Chen},\ and\ \citenamefont
  {Nori}}]{Chen2014a}%
  \BibitemOpen
  \bibfield  {author} {\bibinfo {author} {\bibfnamefont {Y.-N.}\ \bibnamefont
  {Chen}}, \bibinfo {author} {\bibfnamefont {C.-M.}\ \bibnamefont {Li}},
  \bibinfo {author} {\bibfnamefont {N.}~\bibnamefont {Lambert}}, \bibinfo
  {author} {\bibfnamefont {S.-L.}\ \bibnamefont {Chen}}, \bibinfo {author}
  {\bibfnamefont {Y.}~\bibnamefont {Ota}}, \bibinfo {author} {\bibfnamefont
  {G.-Y.}\ \bibnamefont {Chen}}, \ and\ \bibinfo {author} {\bibfnamefont
  {F.}~\bibnamefont {Nori}},\ }\bibfield  {title} {\enquote {\bibinfo {title}
  {Temporal steering inequality},}\ }\href@noop {} {\bibfield  {journal}
  {\bibinfo  {journal} {Phys. Rev. A}\ }\textbf {\bibinfo {volume} {89}},\
  \bibinfo {pages} {032112} (\bibinfo {year} {2014})}\BibitemShut {NoStop}%
\bibitem [{\citenamefont {Piani}\ and\ \citenamefont
  {Watrous}(2015)}]{Piani2015a}%
  \BibitemOpen
  \bibfield  {author} {\bibinfo {author} {\bibfnamefont {M.}~\bibnamefont
  {Piani}}\ and\ \bibinfo {author} {\bibfnamefont {J.}~\bibnamefont
  {Watrous}},\ }\bibfield  {title} {\enquote {\bibinfo {title} {Necessary and
  sufficient quantum information characterization of {Einstein-Podolsky-Rosen}
  steering},}\ }\href@noop {} {\bibfield  {journal} {\bibinfo  {journal} {Phys.
  Rev. Lett.}\ }\textbf {\bibinfo {volume} {114}},\ \bibinfo {pages} {060404}
  (\bibinfo {year} {2015})}\BibitemShut {NoStop}%
\bibitem [{\citenamefont {Skrzypczyk}\ \emph {et~al.}(2014)\citenamefont
  {Skrzypczyk}, \citenamefont {Navascu\'{e}s},\ and\ \citenamefont
  {Cavalcanti}}]{Skrzypczyk2014a}%
  \BibitemOpen
  \bibfield  {author} {\bibinfo {author} {\bibfnamefont {P.}~\bibnamefont
  {Skrzypczyk}}, \bibinfo {author} {\bibfnamefont {M.}~\bibnamefont
  {Navascu\'{e}s}}, \ and\ \bibinfo {author} {\bibfnamefont {D.}~\bibnamefont
  {Cavalcanti}},\ }\bibfield  {title} {\enquote {\bibinfo {title} {Quantifying
  {Einstein-Podonsky-Rosen} steering},}\ }\href@noop {} {\bibfield  {journal}
  {\bibinfo  {journal} {Phys. Rev. Lett.}\ }\textbf {\bibinfo {volume} {112}},\
  \bibinfo {pages} {180404} (\bibinfo {year} {2014})}\BibitemShut {NoStop}%
\bibitem [{\citenamefont {Kadison}\ and\ \citenamefont
  {Ringrose}(1983)}]{Kadison1990a}%
  \BibitemOpen
  \bibfield  {author} {\bibinfo {author} {\bibfnamefont {R.~V.}\ \bibnamefont
  {Kadison}}\ and\ \bibinfo {author} {\bibfnamefont {J.~R.}\ \bibnamefont
  {Ringrose}},\ }\href@noop {} {\emph {\bibinfo {title} {Fundamentals of the
  theory of operator algebras}}}\ (\bibinfo  {publisher} {American Mathematical
  Society},\ \bibinfo {year} {1983})\BibitemShut {NoStop}%
\bibitem [{\citenamefont {Nielsen}\ and\ \citenamefont
  {Chuang}(2010)}]{Nielsen2010a}%
  \BibitemOpen
  \bibfield  {author} {\bibinfo {author} {\bibfnamefont {M.~A.}\ \bibnamefont
  {Nielsen}}\ and\ \bibinfo {author} {\bibfnamefont {I.~L.}\ \bibnamefont
  {Chuang}},\ }\href@noop {} {\emph {\bibinfo {title} {Quantum computation and
  quantum information}}}\ (\bibinfo  {publisher} {Cambridge University Press},\
  \bibinfo {year} {2010})\BibitemShut {NoStop}%
\bibitem [{\citenamefont {Jevtic}\ \emph {et~al.}(2014)\citenamefont {Jevtic},
  \citenamefont {Pusey}, \citenamefont {Jennings},\ and\ \citenamefont
  {Rudolph}}]{Jevtic2014a}%
  \BibitemOpen
  \bibfield  {author} {\bibinfo {author} {\bibfnamefont {S.}~\bibnamefont
  {Jevtic}}, \bibinfo {author} {\bibfnamefont {M.}~\bibnamefont {Pusey}},
  \bibinfo {author} {\bibfnamefont {D.}~\bibnamefont {Jennings}}, \ and\
  \bibinfo {author} {\bibfnamefont {T.}~\bibnamefont {Rudolph}},\ }\bibfield
  {title} {\enquote {\bibinfo {title} {Quantum steering ellipsoids},}\
  }\href@noop {} {\bibfield  {journal} {\bibinfo  {journal} {Phys. Rev. Lett.}\
  }\textbf {\bibinfo {volume} {113}},\ \bibinfo {pages} {020402} (\bibinfo
  {year} {2014})}\BibitemShut {NoStop}%
\bibitem [{\citenamefont {Bengtsson}\ and\ \citenamefont
  {\'{Z}yczkowski}(2006)}]{Bengtsson2006a}%
  \BibitemOpen
  \bibfield  {author} {\bibinfo {author} {\bibfnamefont {I.}~\bibnamefont
  {Bengtsson}}\ and\ \bibinfo {author} {\bibfnamefont {K.}~\bibnamefont
  {\'{Z}yczkowski}},\ }\href@noop {} {\emph {\bibinfo {title} {Geometry of
  quantum states: an introduction to quantum entanglement}}}\ (\bibinfo
  {publisher} {Cambridge University Press},\ \bibinfo {year}
  {2006})\BibitemShut {NoStop}%
\bibitem [{\citenamefont {Moroder}\ \emph {et~al.}(2016)\citenamefont
  {Moroder}, \citenamefont {Gittsovich}, \citenamefont {Huber}, \citenamefont
  {Uola},\ and\ \citenamefont {G{\"{u}}hne}}]{Moroder2014a}%
  \BibitemOpen
  \bibfield  {author} {\bibinfo {author} {\bibfnamefont {T.}~\bibnamefont
  {Moroder}}, \bibinfo {author} {\bibfnamefont {O.}~\bibnamefont {Gittsovich}},
  \bibinfo {author} {\bibfnamefont {M.}~\bibnamefont {Huber}}, \bibinfo
  {author} {\bibfnamefont {R.}~\bibnamefont {Uola}}, \ and\ \bibinfo {author}
  {\bibfnamefont {O.}~\bibnamefont {G{\"{u}}hne}},\ }\bibfield  {title}
  {\enquote {\bibinfo {title} {Steering maps and their application to
  dimensional-bounded steering},}\ }\href@noop {} {\bibfield  {journal}
  {\bibinfo  {journal} {Phys. Rev. Lett.}\ }\textbf {\bibinfo {volume} {116}},\
  \bibinfo {pages} {090403} (\bibinfo {year} {2016})}\BibitemShut {NoStop}%
\bibitem [{\citenamefont {de~Pillis}(1967)}]{Pillis1967a}%
  \BibitemOpen
  \bibfield  {author} {\bibinfo {author} {\bibfnamefont {J.}~\bibnamefont
  {de~Pillis}},\ }\bibfield  {title} {\enquote {\bibinfo {title} {Linear
  transformations which preserve trace and positive semidefiniteness of
  operators},}\ }\href@noop {} {\bibfield  {journal} {\bibinfo  {journal}
  {Pacific J. Math.}\ }\textbf {\bibinfo {volume} {23}},\ \bibinfo {pages}
  {129} (\bibinfo {year} {1967})}\BibitemShut {NoStop}%
\bibitem [{\citenamefont {Jamio{\l}kowski}(1972)}]{Jamiokowski1972a}%
  \BibitemOpen
  \bibfield  {author} {\bibinfo {author} {\bibfnamefont {A.}~\bibnamefont
  {Jamio{\l}kowski}},\ }\bibfield  {title} {\enquote {\bibinfo {title} {Linear
  transformations which preserve trace and positive semidefiniteness of
  operators},}\ }\href@noop {} {\bibfield  {journal} {\bibinfo  {journal} {Rep.
  Math. Phys.}\ }\textbf {\bibinfo {volume} {3}},\ \bibinfo {pages} {275}
  (\bibinfo {year} {1972})}\BibitemShut {NoStop}%
\bibitem [{\citenamefont {Jiang}\ \emph {et~al.}(2013)\citenamefont {Jiang},
  \citenamefont {Luo},\ and\ \citenamefont {Fu}}]{Jiang2013a}%
  \BibitemOpen
  \bibfield  {author} {\bibinfo {author} {\bibfnamefont {M.}~\bibnamefont
  {Jiang}}, \bibinfo {author} {\bibfnamefont {S.}~\bibnamefont {Luo}}, \ and\
  \bibinfo {author} {\bibfnamefont {S.}~\bibnamefont {Fu}},\ }\bibfield
  {title} {\enquote {\bibinfo {title} {Channel-state duality},}\ }\href@noop {}
  {\bibfield  {journal} {\bibinfo  {journal} {Phys. Rev. A}\ }\textbf {\bibinfo
  {volume} {87}},\ \bibinfo {pages} {022310} (\bibinfo {year}
  {2013})}\BibitemShut {NoStop}%
\bibitem [{\citenamefont {Kadison}\ and\ \citenamefont
  {Kromann}(1995)}]{Kadison1995a}%
  \BibitemOpen
  \bibfield  {author} {\bibinfo {author} {\bibfnamefont {L.}~\bibnamefont
  {Kadison}}\ and\ \bibinfo {author} {\bibfnamefont {M.~T.}\ \bibnamefont
  {Kromann}},\ }\href@noop {} {\emph {\bibinfo {title} {Projective geometry and
  modern algebra}}}\ (\bibinfo  {publisher} {Birkh{\"a}user},\ \bibinfo {year}
  {1995})\BibitemShut {NoStop}%
\bibitem [{\citenamefont {Milne}\ \emph {et~al.}(2015)\citenamefont {Milne},
  \citenamefont {Jennings},\ and\ \citenamefont {Rudolph}}]{Milne2015a}%
  \BibitemOpen
  \bibfield  {author} {\bibinfo {author} {\bibfnamefont {A.}~\bibnamefont
  {Milne}}, \bibinfo {author} {\bibfnamefont {D.}~\bibnamefont {Jennings}}, \
  and\ \bibinfo {author} {\bibfnamefont {T.}~\bibnamefont {Rudolph}},\
  }\bibfield  {title} {\enquote {\bibinfo {title} {Geometric representation of
  two-qubit entanglement witnesses},}\ }\href@noop {} {\bibfield  {journal}
  {\bibinfo  {journal} {Phys. Rev. A}\ }\textbf {\bibinfo {volume} {92}},\
  \bibinfo {pages} {012311} (\bibinfo {year} {2015})}\BibitemShut {NoStop}%
\bibitem [{\citenamefont {Jevtic}\ \emph {et~al.}(2015)\citenamefont {Jevtic},
  \citenamefont {Hall}, \citenamefont {Anderson}, \citenamefont {Zwierz},\ and\
  \citenamefont {Wiseman}}]{Jevtic2015a}%
  \BibitemOpen
  \bibfield  {author} {\bibinfo {author} {\bibfnamefont {S.}~\bibnamefont
  {Jevtic}}, \bibinfo {author} {\bibfnamefont {M.~J.~W.}\ \bibnamefont {Hall}},
  \bibinfo {author} {\bibfnamefont {M.~R.}\ \bibnamefont {Anderson}}, \bibinfo
  {author} {\bibfnamefont {M.}~\bibnamefont {Zwierz}}, \ and\ \bibinfo {author}
  {\bibfnamefont {H.~M.}\ \bibnamefont {Wiseman}},\ }\bibfield  {title}
  {\enquote {\bibinfo {title} {{Einstein-Podolsky-Rosen} steering and the
  steering ellipsoid},}\ }\href@noop {} {\bibfield  {journal} {\bibinfo
  {journal} {J. Opt. Soc. Am. B}\ }\textbf {\bibinfo {volume} {32}},\ \bibinfo
  {pages} {A40} (\bibinfo {year} {2015})}\BibitemShut {NoStop}%
\bibitem [{\citenamefont {Rockafellar}(1970)}]{Rockafellar1970a}%
  \BibitemOpen
  \bibfield  {author} {\bibinfo {author} {\bibfnamefont {R.~T.}\ \bibnamefont
  {Rockafellar}},\ }\href@noop {} {\emph {\bibinfo {title} {Convex analysis}}}\
  (\bibinfo  {publisher} {Princeton University Press},\ \bibinfo {year}
  {1970})\BibitemShut {NoStop}%
\bibitem [{\citenamefont {Sanpera}\ \emph {et~al.}(1998)\citenamefont
  {Sanpera}, \citenamefont {Tarrach},\ and\ \citenamefont
  {Vidal}}]{Sanpera1998a}%
  \BibitemOpen
  \bibfield  {author} {\bibinfo {author} {\bibfnamefont {A.}~\bibnamefont
  {Sanpera}}, \bibinfo {author} {\bibfnamefont {R.}~\bibnamefont {Tarrach}}, \
  and\ \bibinfo {author} {\bibfnamefont {G.}~\bibnamefont {Vidal}},\ }\bibfield
   {title} {\enquote {\bibinfo {title} {Local description of quantum
  inseparability},}\ }\href@noop {} {\bibfield  {journal} {\bibinfo  {journal}
  {Phys. Rev. A}\ }\textbf {\bibinfo {volume} {58}},\ \bibinfo {pages} {826}
  (\bibinfo {year} {1998})}\BibitemShut {NoStop}%
\bibitem [{\citenamefont {Werner}(2014)}]{Werner2014a}%
  \BibitemOpen
  \bibfield  {author} {\bibinfo {author} {\bibfnamefont {R.~F.}\ \bibnamefont
  {Werner}},\ }\bibfield  {title} {\enquote {\bibinfo {title} {Steering, or
  maybe why {Einstein} did not go all the way to {Bell's} argument},}\
  }\href@noop {} {\bibfield  {journal} {\bibinfo  {journal} {J. Phys. A: Math.
  Theor.}\ }\textbf {\bibinfo {volume} {47}},\ \bibinfo {pages} {424008}
  (\bibinfo {year} {2014})}\BibitemShut {NoStop}%
\bibitem [{\citenamefont {Bowles}\ \emph {et~al.}(2016)\citenamefont {Bowles},
  \citenamefont {Hirsch}, \citenamefont {Quintino},\ and\ \citenamefont
  {Brunner}}]{Bowles2016a}%
  \BibitemOpen
  \bibfield  {author} {\bibinfo {author} {\bibfnamefont {J.}~\bibnamefont
  {Bowles}}, \bibinfo {author} {\bibfnamefont {F.}~\bibnamefont {Hirsch}},
  \bibinfo {author} {\bibfnamefont {M.~T.}\ \bibnamefont {Quintino}}, \ and\
  \bibinfo {author} {\bibfnamefont {N.}~\bibnamefont {Brunner}},\ }\bibfield
  {title} {\enquote {\bibinfo {title} {Sufficient criterion for guaranteeing
  that a two-qubit state is unsteerable},}\ }\href@noop {} {\bibfield
  {journal} {\bibinfo  {journal} {Phys. Rev. A}\ }\textbf {\bibinfo {volume}
  {93}},\ \bibinfo {pages} {022121} (\bibinfo {year} {2016})}\BibitemShut
  {NoStop}%
\bibitem [{\citenamefont {Nguyen}\ and\ \citenamefont
  {Vu}(2016)}]{Nguyen2016b}%
  \BibitemOpen
  \bibfield  {author} {\bibinfo {author} {\bibfnamefont {H.~C.}\ \bibnamefont
  {Nguyen}}\ and\ \bibinfo {author} {\bibfnamefont {T.}~\bibnamefont {Vu}},\
  }\bibfield  {title} {\enquote {\bibinfo {title} {Necessary and sufficient
  condition for steerability of two-qubit states by the geometry of steering
  outcomes},}\ }\href@noop {} {\bibfield  {journal} {\bibinfo  {journal}
  {arXiv:1604.03815}\ } (\bibinfo {year} {2016})}\BibitemShut {NoStop}%
\end{thebibliography}%
\end{document}